\documentclass[submission,copyright,creativecommons]{eptcs}
\providecommand{\event}{TARK 2021} 
\usepackage[utf8]{inputenc}
\usepackage{url,hyperref}
\usepackage{amsmath}
\usepackage{amssymb}
\usepackage{amsfonts}
\usepackage{amssymb,textcomp}
\usepackage{amsthm,tasks}
\usepackage[]{algorithm2e}
\usepackage{float}
\usepackage{tikz-cd}
\usepackage{wrapfig}
\usetikzlibrary{positioning,automata,arrows}

\usepackage{graphicx}

\usepackage{wrapfig}

\usepackage{calc}

\newtheorem{theorem}{Theorem}
\newtheorem{lemma}{Lemma}
\newtheorem{corollary}{Corollary}
\newtheorem{definition}{Definition}

\usepackage{calc} 

\usepackage{color}

\title{Communication Pattern Models: An Extension of Action Models for Dynamic-Network Distributed Systems}

\author{
Diego A.\ Vel\'azquez\thanks{Diego A.\ Vel\'azquez is a doctoral student at the Programa de Doctorado en Ciencia e Ingenier\'ia de la Computaci\'on, Universidad Nacional Aut\'onoma de M\'exico, and is the receipient of a fellowship from CONACyT.}
\email{velazquez-diego@ciencias.unam.mx}
\and
Armando Casta\~neda\thanks{Armando Casta\~neda was supported by PAPIIT project IN108720.}
\email{armando.castaneda@im.unam.mx}
\and
David A.\ Rosenblueth
\email{drosenbl@unam.mx}
\and
Universidad Nacional Aut\'onoma de M\'exico
}

\def\titlerunning{Communication Pattern Models: An Ext. of Action Models for Dynamic-Network Dist. Systems}
\def\authorrunning{D. A. Vel\'azquez, A. Casta\~neda \& D. A.\ Rosenblueth }

\newcommand{\I}[1]{\mathit{#1}}

\newcommand{\mc}{\ensuremath{\mathit{,}}}

\newcommand{\kmf}[1]{\mathit{#1}}
\newcommand{\km}{\ensuremath{\kmf{M}}}
\newcommand{\amf}[1]{\mathit{#1}}
\newcommand{\am}{\ensuremath{\amf{A}}}
\newcommand{\cpmf}[1]{\mathit{#1}}
\newcommand{\cpm}{\ensuremath{\mathcal{P}}}

\begin{document}

\maketitle

\begin{abstract}
  Halpern and Moses were the first to recognize, in 1984, the importance of a formal treatment of knowledge in distributed computing.
  Many works in distributed computing, however, still employ informal notions of knowledge.
  Hence, it is critical to further study such formalizations.
  Action models, a significant approach to modeling dynamic epistemic logic, have only recently been applied to distributed computing,
for instance, by Goubault, Ledent, and Rajsbaum.
  Using action models for analyzing distributed-computing environments, as proposed by these authors, has drawbacks, however.
  In particular, a direct use of action models may cause such models to grow exponentially as the computation of the distributed system evolves.
Hence, our motivation is finding compact action models for distributed systems.
We introduce \emph{communication pattern models} as an extension of both ordinary action models and their update operator. 
We give a systematic construction of communication pattern models 
for a large variety of distributed-computing models called \emph{dynamic-network models}. 
For a proper subclass of dynamic-network models called \emph{oblivious},
the communication pattern model remains the same throughout the computation.
\end{abstract}


\section{Introduction}
A formal treatment of the concept of \emph{knowledge} is important yet little studied in the distributed-computing literature.
Authors in distributed computing often refer to the knowledge of the different agents or processes, 
but typically do so only \emph{informally}.
Hence, a formal basis of knowledge in distributed computing would increase the power to prove formal results.
The first step in this direction was taken by Halpern and Moses~\cite{halpern-moses-84}.
A topological approach~\cite{herlihy-kozlov-rajsbaum-14} to distributed computing resulted in a further connection~\cite{goubault-ledent-rajsbaum-18} with epistemic logic.
Such a connection uses epistemic ``action models''~\cite{baltag-lawrence-slawomir-98,van-ditmarsch-van-der-hoek-kooi-08} to capture communication between agents.
We observe, however, that the action models proposed in~\cite{goubault-ledent-rajsbaum-18} for the \emph{Iterated Immediate Shapshot} (IIS)
computing model~\cite{borowsky-gafni-93-2, borowsky-1997} not only vary at each communication round, but each action model itself is structurally isomorphic to the resulting epistemic model, which
paradoxically, requires knowing the desired result beforehand.
Our objective is to develop a different connection between 
\emph{dynamic-network models}~\cite{charron-bost-schiper09,kuhn-oshman-11,nowak-schmid-winkler-19} (which include the IIS model),
and \emph{Dynamic Epistemic Logic} (DEL)~\cite{baltag-lawrence-slawomir-98,van-ditmarsch-van-der-hoek-kooi-08},
more appropriate for computing knowledge change in these systems.
\vspace{-10pt}
\paragraph{Context.}
In a distributed system, communication is typically performed either by sending and receiving messages, or by writing to, and reading from, a shared memory.
The communication patterns (i.e., who communicated with whom) that can occur may change from model to model.
When designing and analyzing distributed systems, it is often the case that authors informally refer to what an agent ``knows'' after an agent performs some action.
There is indeed a formal connection between distributed systems and epistemic logic: 
this connection was initiated by Halpern and Moses in 1984~\cite{halpern-moses-84}, showing that distributed systems can be rigorously studied from 
an epistemic-logic viewpoint.
Roughly, a distributed protocol is studied through an \emph{epistemic model} with each of its states 
representing
a possible \emph{configuration} of the protocol.
Since its discovery, the epistemic-based approach to distributed systems has been fruitful,
as shown in the book by Fagin, Halpern, Moses, and Vardi~\cite{fagin-halpern-moses-vardi-95}.

An important connection between distributed computing and topology was discovered in three independent papers by Borowsky and Gafni~\cite{borowsky-gafni-93}, Herlihy and Shavit~\cite{herlihy-shavit-93}, and Saks and Zaharoglou~\cite{saks-zaharoglou-93} in 1993, and since then this approach has provided useful techniques to show a number of important results in this field.
The book by Herlihy, Kozlov, and Rajsbaum~\cite{herlihy-kozlov-rajsbaum-14} provides a comprehensive description of this connection.

Recently, Goubault, Ledent, and Rajsbaum have shown~\cite{goubault-ledent-rajsbaum-18} that the epistemic-based approach can be directly connected to the topology-based approach to distributed systems.
The topological approach studies a distributed protocol through its topological representation: a geometric object, called \emph{simplicial complex}, where each of its faces is associated with a \emph{configuration} of the protocol.
In essence, Goubault, Ledent, and Rajsbaum established~\cite{goubault-ledent-rajsbaum-18} a correspondence between the topological description of distributed protocols and epistemic models.

A second interesting result of these authors is that the communication patterns allowed in the IIS distributed model can also be described using epistemic-logic tools from DEL:
the communication in a distributed model can be modeled with an action model capturing the communication events that can occur, and the \emph{restricted modal product} operator shows how knowledge evolves after agents exchange information in a \emph{communication round}.
(A more thorough summary of~\cite{goubault-ledent-rajsbaum-18} appears in Section~\ref{related-work}.)



We observe that the action models of~\cite{goubault-ledent-rajsbaum-18} describing communication in the IIS model have drawbacks:
First, such action models are different for each communication round (an ideal representation of communication would not depend on the
communication rounds that have been executed so far).
In addition, the size of such action models grows exponentially as the computation develops.
Moreover, such action models are structurally isomorphic to the epistemic models we wish to compute.
The action models of~\cite{goubault-ledent-rajsbaum-18}, therefore,
not only are not useful for computing the epistemic model resulting from a communication event, but
are not a succinct representation of the communication that can happen in the IIS model.
This phenomenon is opposite to the description of communication in the topological approach, where we have a geometric and compact description of the communication in the IIS model: the communication is clearly described as a subdivision~\cite[Chapter 11]{herlihy-kozlov-rajsbaum-14}. 
\footnote{Informally, a subdivision results from dividing the faces of a geometrical object into more faces preserving its shape. }
\vspace{-10pt}
\paragraph{Contributions.}
We are interested in the following question: in the spirit of the action-model approach to DEL, is it possible to describe the communication in a distributed model in a compact manner?
As a first step, we try to salvage the approach of~\cite{goubault-ledent-rajsbaum-18}, by attempting to find an action model applicable to every communication round
for \emph{two} agents with binary inputs in the IIS model.
We exhibit a family of action models with a constant number of events, although each event is labeled with a precondition formula whose size does increase at each communication round.
For obtaining these action models, it was crucial to know \emph{in advance} the epistemic model after a communication round.
We have not been able to find a similar family for \emph{three or more} agents yet.
The case of $m$-ary inputs for $m \geq 3$ would be even harder to analyze.

The drawbacks of the action models proposed in~\cite{goubault-ledent-rajsbaum-18}, together with our unsuccessful efforts to find action models for IIS of small size, are motivations for investigating a different approach.
We hence consider an extension of action models that allows us to easily derive models of small size.
Moreover, we study not only the IIS model but also a larger class of message-passing models called 
\emph{dynamic-network models}~\cite{charron-bost-schiper09,kuhn-oshman-11,nowak-schmid-winkler-19}.
Roughly speaking, in a dynamic-network model, the agents execute infinite sequences of communication rounds.
In each round, the agents communicate according to a \emph{communication pattern} that specifies 
who communicates with whom in that round.
A proper subclass of dynamic-network models are those known as \emph{oblivious} 
that are specified with a set of communication patterns that can occur in any round, 
regardless of the communication patterns that have occurred so far in the execution.
The IIS model can be defined as an oblivious dynamic-network model.



Our main contribution is a simple but powerful extension to the existing action models and its restricted modal product. 
For every dynamic-network model, we systematically define an infinite sequence of \emph{communication pattern models} 
that represent how knowledge changes when agents communicate in
the \emph{full-information protocol},
hence making our approach amenable to be extended to automated formal verification of distributed systems. 
For the case of oblivious models, the communication pattern model remains the same all through the execution. 
Hence, we are able to model communication of oblivious dynamic-network models in constant space.

\paragraph{Structure of this paper.}
The rest of this paper is structured as follows.
Section 2 gives an overview of drawbacks arising from a straightforward use of action models in some contexts in multi-agent systems and outlines our solution to overcome such shortcomings.
After establishing notation and definitions in Section~3, we explain, in Section~4, an attempt to improve on~\cite{goubault-ledent-rajsbaum-18} within the IIS model.
Section 5 presents communication pattern models, our modification of action models.
Comparison with existing work appears in Section~6, and Section~7 concludes this paper.


\section{An Overview of Our Proposal}

We first motivate our proposal by pointing out a limitation
of action models and the restricted modal product that arises in some contexts 
when modeling the communication that can happen in a multi-agent system. 
Roughly speaking, sometimes it is impossible to have action models of ``small size''.
Our discussion here is informal as we are interested in high-level ideas at the moment,
hence delaying formal definitions for the next sections.

\vspace{-0.3cm}
\paragraph{The issue.}
Let us consider the well-known \emph{coordinated attack} problem where two agents 
$a$ and $b$ wish to schedule an attack. 
Agent $a$ has two possible preferences for scheduling the attack, {\sf n} for noon or {\sf d} for dawn,
while agent $b$ has no initial preference and wishes to learn $a$'s. 
Communication is unreliable: whenever an agent sends a message, such a message can get lost.
The epistemic model $\km$ modeling the initial situation before any communication occurs has two worlds,
one in which $a$ prefers to attack at dawn and another one in which $a$ prefers to attack at noon;
$b$ cannot distinguish between these two worlds.
See model $\km$ in Fig.~\ref{fig-coordinated-attack-identified}.

An action model is a generalization of an epistemic model, where vertices, called \emph{events}, are labeled with arbitrary formulas (as opposed to sets of propositional variables) called \emph{preconditions}.
The restricted modal product of an epistemic model $\km$ and an action model $\am$, denoted $\km \otimes \am$, is an epistemic model where each world is a pair $(w,e)$, such that $w$ is a world in $M$, $e$ is an event in $\am$, and the precondition of $e$ holds in $w$. 
Worlds $(w,e)$ and $(w',e')$ are connected with each other for agent $a$ if both $w$ and $w'$ are connected in $\km$ for $a$, and $e$ and $e'$ are connected in $\am$ for $a$.
World $(w,e)$ is labeled with the same label as that of $w$.
(Formal definitions of action model and restricted modal product appear in Subsect.~\ref{action-models}.)

A simple action model $\am$ modeling that $a$ sends its preference to $b$ has 
three events: one for each preference $p \in \{{\sf d, n}\}$ modeling that
$b$ successfully receives the preference, $p$, of $a$,
with a precondition specifying that the event can happen only if $p$ is the preference of $a$,
and a third event, modeling that $a$'s message gets lost, with precondition $\top$. 
Agent $b$ can distinguish between all events since it either receives $a$'s message or not,
but $a$ cannot distinguish between events because messages can get lost.
See action model $\am$ in Fig.~\ref{fig-coordinated-attack-identified}.
The restricted modal product
$\km \otimes \am$ contains four worlds, one for each combination of $a$'s initial preference and successful/unsuccessful communication. 

\vspace{-0.2cm}
\begin{figure}[h]
\begin{center}
\includegraphics[scale=0.65]{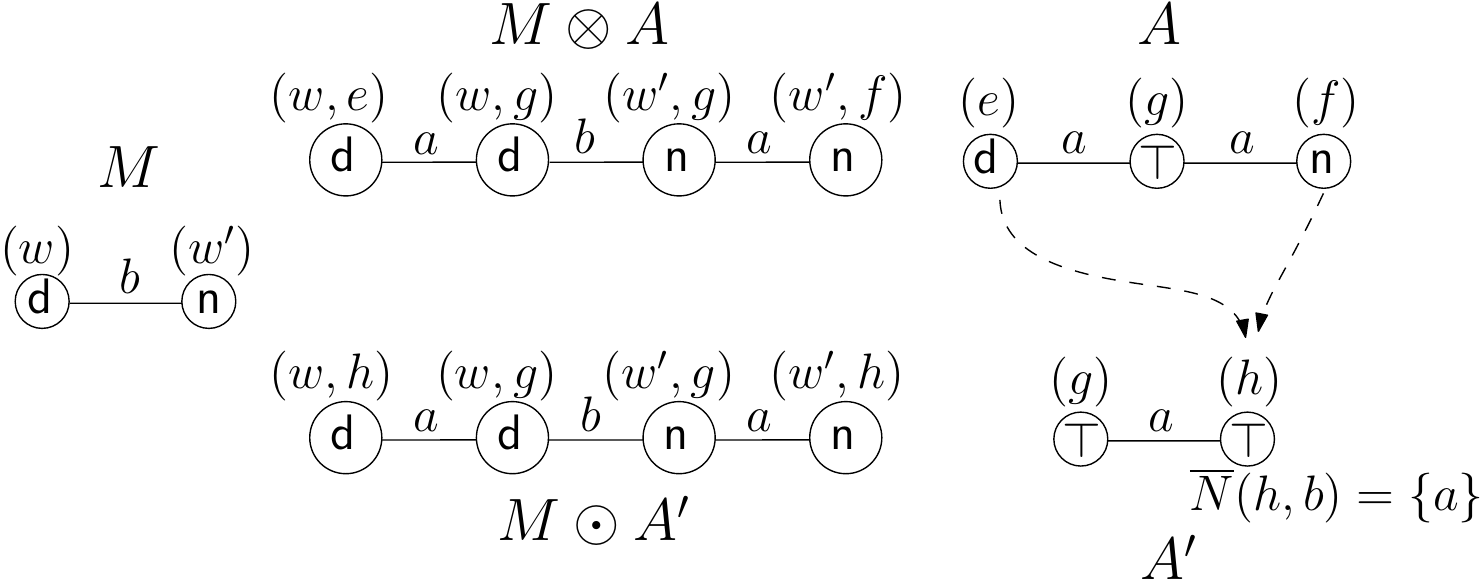}
\caption{\small A smaller action model for the coordinated attack problem using our approach.}
\label{fig-coordinated-attack-identified}
\end{center}
\end{figure}
\vspace{-0.6cm}

We observe that the action model $\am$ has the following inconvenience.
If $a$ has $x > 2$ preferences to schedule the attack instead of only two, 
a natural generalization of $\am$ has $x+1$ events: $\am$ is akin to a star
with a ``central'' event modeling that $a$'s message gets lost, and one
event for each of the possible preferences of $a$.
Thus, the size of the action model is proportional to the size of $a$'s \emph{input space}. 

Can we design a smaller action model for this situation? 
Can we design an action model with only two events, one corresponding to the
case that $a$'s message gets lost and another corresponding to the case
that $a$'s message (with distinct contents, either $\sf d$ or $\sf n$) reaches $b$?
The answer is no.
It is easy to see that if we have an action model $\am'$
with only one event $h$ corresponding to $a$'s successful communication,
and unavoidably with precondition ${\sf d} \vee {\sf n} = { \top}$ 
(see $\am'$ in Fig.~\ref{fig-coordinated-attack-identified}, discarding for the
moment the set $\overline{N}(h,b) = \{a\}$),
then $\km \otimes \am'$ has again four worlds but now $b$ 
cannot distinguish 
between the worlds $(w, h)$ and $(w', h)$ corresponding to the cases where 
the communication was successful (which is incorrect).
A similar situation happens if $a$ has more than two preferences: it is impossible
to have an action model with an event that models the case that $a$'s communication is successful.
We cannot get any smaller action model in this situation (it might be possible, however, to do so in further rounds)
because the action models are designed to deal with ``interpreted'' events, namely, 
an event includes the 
information of the message encoded in its precondition, 
hence it has a limited ability to represent that \emph{some} information is sent from one agent to another.

This property of action models is a problem in some contexts.
Specifically, when studying computability in a given distributed model, 
it is often the case that the analysis is performed on \emph{protocols} 
in which agents proceed in a sequence of \emph{rounds} of communication, 
and in each round every agent 
sends \emph{all} the information it has collected so far 
to all other agents; these protocols are called \emph{full-information} in the distributed-computing literature.
To be able to reason in the style of DEL, we would like to have an action model modeling 
the communication events that can happen in a round,
and update the epistemic model with the help of the restricted modal product in each round.
Drawbacks of the approach in~\cite{goubault-ledent-rajsbaum-18} are that the size of the action models proposed there grows exponentially in the number of round and 
that such action models are structurally isomorphic to the resulting epistemic models. 
As we will see later, for the case of two agents, we have been able to find a family of action models with a constant number of actions (although the preconditions of the action model do change from round to round) but it is unclear how to find action models with this property for other cases. 

\paragraph{A glimpse of our solution.}
Coming back to our initial example, how do we ``fix'' the problem
in $M \otimes {\am'}$, i.e., that $b$ cannot distinguish between worlds  $(w, h)$ and $(w', h)$?
Our solution is based on the following observation: 
$b$ must be able to distinguish between the two worlds because
(1) it receives a message from $a$ in the event $h$ of $\am'$, and
(2) $a$ can distinguish between the $w$ and $w'$ in $M$.
Therefore, $a$ must send information that makes $b$ able to distinguish 
between the two worlds in $M \otimes \am'$.

We define an extension to the action model formalism, which equips an action model with an additional
function $\overline{N}$ that maps every pair $(e,a)$ to a set of agents. 
Intuitively, $\overline{N}(e,a)$ contains the agents that $a$ receives messages from
when the event $e$ happens. 
The restricted modal product is modified by adding 
two conditions when updating the accessibility relation of an epistemic model.
Such conditions say that an agent $a$ cannot distinguish between two
worlds $(w,e)$ and $(w',e')$ if and only if $a$ receives
messages from the same set of agents in $e$ and $e'$ (i.e.,
$\overline{N}(e,a) = \overline{N}(e',a)$) and each of these agents cannot distinguish
between  $w$ and $w'$ (namely, $\forall a' \in \overline{N}(e,a), w\sim_{a'} w'$).
The idea is that if those agents sending information to $a$ cannot distinguish between $w$ and $w'$,
then there is no information they send to $a$ making $(w,e)$ and $(w', e')$
distinguishable to $a$. 
The new product is denoted $\odot$.
Using this formalism, for the coordinated attack problem, we are able to define 
a \emph{communication pattern model} $A'$ with a \emph{single} event (called communication pattern in our context) $h$ corresponding 
to the case in which $a$'s message reaches $b$.
In $A'$, $\overline{N}(h,b)$ is set to $\{a\}$, and $\overline{N}$ is set to $\emptyset$ in any other case.
Figure~\ref{fig-coordinated-attack-identified} shows the model $A'$.
Furthermore, the action model is correct regardless of the size of $a$'s input space,
meaning that the very same action model produces the desired epistemic model 
if $M$ represents the situation that $a$ has $x > 2$ initial preferences.


\vspace{-0.4cm}
\section{Analyzing distributed computing models}
\label{sec-dc-models}
\vspace{-0.2cm}

In this section, we give some introductory definitions and fix the notation.
We assume some familiarity with basic epistemic logic. We refer to the language of multiagent epistemic logic as $\mathcal{L}_K$. Additionally, we consider a non-empty finite set of agents $\I{Ag}=\{a_1,\dots,a_n\}$ and a non-empty finite set of propositions $\I{Props}$, unless specified otherwise.
\paragraph{Our models of interest.}
We are interested in \emph{dynamic-network models}~\cite{charron-bost-schiper09,kuhn-oshman-11,nowak-schmid-winkler-19},
in which a set of $n \geq 2$ failure-free agents proceed in an infinite sequence of \emph{synchronous} rounds of communication.
Each agent is a \emph{state machine}.
In each round, the communication is specified with a \emph{communication graph},
namely a directed graph whose vertex set is $\I{Ag}$, with each edge $(a_i,a_j)$ indicating that a message from $a_i$ to $a_j$ is successfully delivered in that round.
The \emph{in-neighborhood} of an agent $a_i$ in a communication graph $G$, namely the set of agents $a_j$ such that  $(a_j,a_i)$ is an arrow in $G$, is denoted $N^-_G(a_i)$.
Let $CP_{\I{Ag}}$ denote the set with all communication graphs with vertex set $Ag$.
Thus, a dynamic-network model $\I{Adv}$ is specified with a set of infinite sequences of graphs of $CP_{\I{Ag}}$, that we call \emph{adversary}.
Intuitively, we say that an adversary $Adv$ is \emph{oblivious} if in
every round, any communication graph in a given set can happen, regardless of the communication graphs
that have happened in previous rounds.
This is formalized as follows. We say that a finite sequence $S$ of communication graphs is a
\emph{prefix} of $\I{Adv}$ if $S$ is a prefix of a sequence in $\I{Adv}$.
An adversary $\I{Adv}$ is oblivious if there exists a non-empty subset $X \subseteq CP_{\I{Ag}}$ such that
the graphs in $X$ are the prefixes of $\I{Adv}$ of length one, and
for every
finite sequence $S$ that is a prefix of $\I{Adv}$, it holds that $S \cdot G$ is a
prefix of $\I{Adv}$, for every graph $G \in X$.
Thus, an oblivious adversary is
simply specified through the set $X$ of communication graphs; we will say that $Adv = X$.

\paragraph{Protocols.} Each agent locally executes a \emph{protocol}
that specifies the messages that the agent sends in a round, depending on
the local state of the agent at the beginning of the round.
Each agent starts the computation with a private input, which is the state of the agent at the beginning of the first round.
Since we are interested in modeling how knowledge can evolve through
the computation, we assume that in every round
every agent attempts to communicate to everybody all it knows so far.
Formally, every agent locally executes the \emph{full-information}
protocol, namely,
in every round an agent sends to all other agents all the information such an agent has collected so far. 
Therefore, the full-information protocol captures all that an agent can know in an execution.
The full-information protocol is an important tool in distributed-computing computability research.

\paragraph{Executions and configurations.}
An \emph{execution} $E$ of an adversary $\I{Adv}$ is a pair $(I,S^{\infty})$, where
$I = (v_1,v_2,\dots,v_n)$ is an \emph{input vector} denoting that
agent $a_i$ starts with input $v_i$, with $v_i$ belonging to an
\emph{input space}, denoted $In$,
and $S^{\infty}$ is a sequence of $\I{Adv}$.
An \emph{$r$-execution} of $\I{Adv}$ is a pair $(I,S)$, where 
$I$ is an input vector and $S$ 
is a prefix of $Adv$ with $|S| = r$.
A \textit{configuration} $C$ is an $n$-tuple whose $i$-th position is
a local state of agent $a_i$ (thus input vectors are configurations). 
We say that $a_i$ \emph{does not distinguish} between configurations
$C$ and $C'$ if and only if $C(i) = C'(i)$.
An $r$-execution $(I,S)$ \emph{ends} at a configuration $C$
if each agent $a_i$ has the local state $C(i)$
after the execution of the sequence of communication rounds described by $S$
with the inputs stated by $I$; alternatively, we say that $C$ is the configuration at the \emph{end} of $(I,S)$.
Note that for the empty sequence, denoted $[\,]$, $I$ is the configuration at the end of the $0$-execution
$(I,[\,])$, for every input vector $I$.

\paragraph{Our representation.}

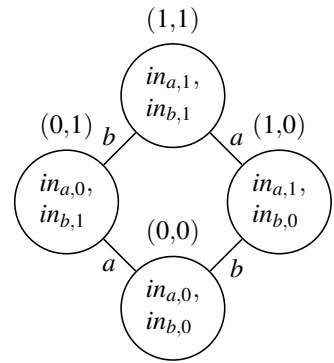
\begin{wrapfigure}{r}{0.3\textwidth}
\vspace{-35pt}
    \small
    \centering
    \begin{tikzpicture}[shorten >=0pt,node distance=2cm,on grid,auto,semithick]
        \node[align=left,state,label=$(0\mc 1)$]            (q_0)                   {$\I{in}_{a,0},$\\$\I{in}_{b,1}$};
        \node[align=left,state,label=$(1\mc 1)$]            (q_1)  [above right=of q_0]    {$\I{in}_{a,1},$\\$\I{in}_{b,1}$};
        \node[align=left,state,label=$(0 \mc 0)$]            (q_2)  [below right=of q_0] {$\I{in}_{a,0},$\\$\I{in}_{b,0}$};
        \node[align=left,state,label=$(1 \mc 0)$]            (q_3)  [below right=of q_1] {$\I{in}_{a,1},$\\$\I{in}_{b,0}$};
        \path (q_0) edge              node[above left=-3pt]        {$b$} (q_1)
                        edge              node [below left=-3pt] {$a$} (q_2)
                    (q_1) edge              node[above right=-3pt]        {$a$} (q_3)
                    (q_2) edge              node [below right=-3pt] {$b$} (q_3);
    \end{tikzpicture}
        \caption{Model $M^0$ for agents $a$ and $b$ with binary inputs.}
        \label{fig:ex0}
\vspace{-15pt}
\end{wrapfigure}
We use \emph{epistemic models}, that we name $M^r$, for representing the $r$-executions of a given adversary $\I{Adv}$.
%
%
An \emph{epistemic model} for $\I{Ag}$ and a set of propositions $\I{Props}$
is a triple $M =$ \mbox{$(W,\sim,L)$}, where $W$ is a finite set of worlds, 
$\sim\;: \I{Ag} \rightarrow \wp ( W \times W)$ assigns an equivalence relation to each agent,
and $L: W \rightarrow \wp ({\it Props})$ assigns a set of true-valued propositions to each world.
Each world in $M^r$ represents an $r$-execution and the accessibility relations represent the indistinguishability relations over the configurations at the end of the $r$-executions of $\I{Adv}$.
\paragraph{The initial epistemic model ($M^0$).}

%
%

We build the initial epistemic model $M^0=(W^0,\sim^0,L^0)$ for $\I{Ag}$ and $\I{In}$ 
with $\I{Props}=\{in_{a,v} \mid a \in \I{Ag} \land v \in \I{In} \}$ so that 
$W^0=\{I \mid I \textrm{ is an input vector for $\I{Ag}$ and $\I{In}$} \}$, 
$I\sim^0_{\I{a}_i}I'$ if and only if $I(i)=I'(i)$, and
$L(I)=\{in_{a_i,v} \in \I{Props} \mid I(i)=v \}$.
The epistemic model $M^0$ for the agents $\I{Ag}=\{a,b\}$ and binary inputs $\I{In}=\{0,1\}$ is depicted in Fig.~\ref{fig:ex0}.



\vspace{-0.3cm}
\section{Action models and the IIS model}
\vspace{-0.2cm}

In this section, we first present the IIS model.
Next, we give the definition of action models.
Finally, we exhibit our best action-model solution of modeling IIS for agents $a$ and $b$ with binary inputs.

\subsection{Iterated Immediate Snapshot distributed-computing model} \label{sec:IIS}

\begin{figure}[ht]
\small
    \centering
\begin{tikzpicture}[shorten >=0pt,node distance=2cm,on grid,auto,semithick]
    \node[align=left,state,label=$(0\mc \{b\}\{a\})$]       (q_0)                   {$\I{in}_{a,0},$\\$\I{in}_{b,0}$};
    \node[align=left,state,label=$(0\mc \{a\mc b\})$]       (q_1)   [right=of q_0]  {$\I{in}_{a,0},$\\$\I{in}_{b,0}$};
    \node[align=left,state,label=$(0\mc \{a\}\{b\})$]       (q_2)   [right=of q_1]  {$\I{in}_{a,0},$\\$\I{in}_{b,0}$};
    \node[align=left,state,label=$(1\mc \{a\}\{b\})$]       (q_3)   [right=of q_2]  {$\I{in}_{a,0},$\\$\I{in}_{b,1}$};
    \node[align=left,state,label=$(1\mc \{a\mc b\})$]       (q_4)   [right=of q_3]  {$\I{in}_{a,0},$\\$\I{in}_{b,1}$};
    \node[align=left,state,label=$(1\mc \{b\}\{a\})$]       (q_5)   [right=of q_4]  {$\I{in}_{a,0},$\\$\I{in}_{b,1}$};
    \node[align=left,state,label=below:$(3\mc \{b\}\{a\})$] (q_6)   [below=of q_5]  {$\I{in}_{a,1},$\\$\I{in}_{b,1}$};
    \node[align=left,state,label=below:$(3\mc \{a\mc b\})$] (q_7)   [left=of q_6]   {$\I{in}_{a,1},$\\$\I{in}_{b,1}$};
    \node[align=left,state,label=below:$(3\mc \{a\}\{b\})$] (q_8)   [left=of q_7]   {$\I{in}_{a,1},$\\$\I{in}_{b,1}$};
    \node[align=left,state,label=below:$(2\mc \{a\}\{b\})$] (q_9)   [left=of q_8]   {$\I{in}_{a,1},$\\$\I{in}_{b,0}$};
    \node[align=left,state,label=below:$(2\mc \{a\mc b\})$] (q_10)  [left=of q_9]   {$\I{in}_{a,1},$\\$\I{in}_{b,0}$};
    \node[align=left,state,label=below:$(2\mc \{b\}\{a\})$] (q_11)  [left=of q_10]  {$\I{in}_{a,1},$\\$\I{in}_{b,0}$};
    \path (q_0)   edge    node    {$a$} (q_1)
            (q_1)   edge    node    {$b$} (q_2)
            (q_2)   edge    node    {$a$} (q_3)
            (q_3)   edge    node    {$b$} (q_4)
            (q_4)   edge    node    {$a$} (q_5)
            (q_5)   edge    node    {$b$} (q_6)
            (q_6)   edge    node    {$a$} (q_7)
            (q_7)   edge    node    {$b$} (q_8)
            (q_8)   edge    node    {$a$} (q_9)
            (q_9)   edge    node    {$b$} (q_10)
            (q_10)  edge    node    {$a$} (q_11)
            (q_11)  edge    node    {$b$} (q_0);
\end{tikzpicture}
    \caption{Epistemic model $M_{\mathrm{IIS}}^1$ that represents the configurations at the end of the first round of the full-information protocol for two agents $a$ and $b$ with binary inputs in the IIS model.}
    \label{fig:kmIIS}
\end{figure}
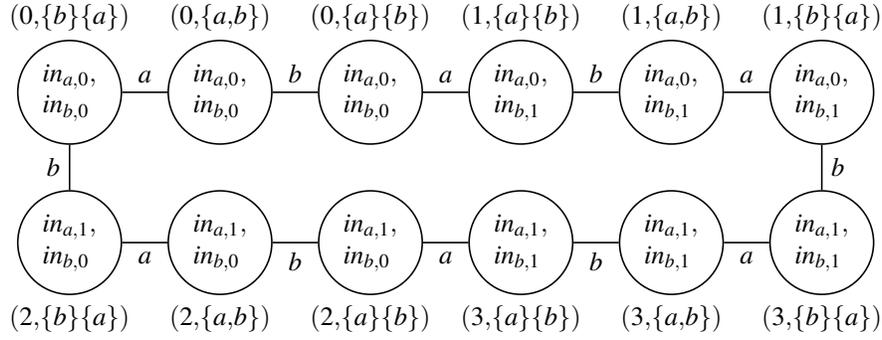

The IIS model~\cite{borowsky-1997} is a fundamental model that
fully captures what can be solved in asynchronous wait-free shared-memory systems with process-crash failures.
We can define IIS as a (failure-free) synchronous oblivious dynamic-network adversary.
The set describing the adversary is as follows.
For every sequence of non-empty subsets of $\I{Ag}$, $S=[C_1,C_2,\dots,C_k]$,  
satisfying that $Ag = \bigcup C_i$ and $C_i \cap C_j = \emptyset$ whenever $i \neq j$,
the adversary has the communication graph with a directed edge $(a,b)$ for every pair of agents $a \in C_i, b \in C_j$ with
$1 \leq i \leq j \leq k$. We say that $C_i$ is a \emph{concurrency class}.
In Fig.~\ref{fig:kmIIS}, we show the epistemic model $M_{\mathrm{IIS}}^1$ that represents the configurations at the end of the first round of the full-information protocol
 for processes $a$ and $b$ with binary input in the IIS model.
\subsection{Action models}
\label{action-models}
%
%
Action models were introduced in~\cite{baltag-lawrence-slawomir-98} as a general way to model dynamics of knowledge via events.
%


\begin{definition}[Action model]
    An \emph{action model} $\am$ is a triple $\amf{(E,R,Pre)}$, where
$\amf{E}$ is a non-empty finite set of events, $\amf{R : P \rightarrow \wp(E \times E) }$ is a function that associates each agent with a relation over the set of events, and $\amf{Pre: E \rightarrow \mathcal{L}_K}$ is a function that associates each event with a precondition.
\end{definition}


%

\begin{definition}[Syntax]
    Let $\am\amf{=(E,R,Pre)}$ be an action model over $\I{Ag}$ and $\mathit{Props}$.
    The language $\mathcal{L}_{\otimes}$ is given by the following \textit{BNF}
    $\varphi ::= p \mid \lnot \varphi \mid~\varphi \land \varphi~|~K_a\varphi~|~[(\am ,\amf{e})]\varphi$,
where $a \in \I{Ag}$, and $p \in \I{Props}$, $\amf{e} \in \amf{E}$ and $(\am ,\amf{e})$ is an \textit{update}.
\end{definition}

%


\begin{definition}[Restricted modal product]
Let $\km=(W,\sim,L)$ be an epistemic model over $\I{Ag}$ and $\mathit{Props}$.
Let $\am\amf{=(E,R,Pre)}$ be an action model.
$\km'=\kmf{(W',\sim',L')}=\km \otimes \am $ is defined as follows:

\begin{itemize}
    \item $\kmf{W}' = \{ (\kmf{w},\amf{e}) \in  \kmf{W} \times \amf{E} \mid M,w \models \amf{Pre(e)} \} $
    \item $\sim'_a\; = \{((\kmf{w},\amf{e}),(\kmf{w}',\amf{e}')) \in \kmf{W}' \times \kmf{W}' \mid w\sim_aw' \;\land\; e\;R_a\;e'\}$ 
    \item $\kmf{L}'((w,e))=\kmf{L}(w)$
\end{itemize} 
\end{definition}

\newcommand{\nsp}{\hspace*{-80pt}}

%
\begin{definition}[Semantics]
Let $M=(W,R,L)$ be an epistemic model over $\I{Ag}$ and $\mathit{Props}$. 
Let $\am\amf{=(E,R,\I{Pre})}$ be an action model.
Let $p \in \mathit{Props}$ be a proposition.
Let $w,w' \in W$ be worlds. 
Let $a \in \I{Ag}$ be an agent.
Let $\amf{e \in E}$ be an event.
Let $\varphi,\psi \in \mathcal{L}_{\otimes}$ be formulas.
\begin{align*}
M,w & \models p\nsp&\textrm{ iff }&~p \in L(w)\\
M,w & \models \lnot \varphi\nsp&\textrm{iff }&M,w \not\models  \varphi\\
M,w & \models \varphi \land \psi\nsp&\textrm{iff }&M,w \models \varphi~\textrm{and}~ M,w \models \psi\\
M,w & \models K_p\varphi\nsp&\textrm{iff }&M,w' \models \varphi~\textrm{for~all~$w'$~such~that }~w\;R(p)\;w'\\
M,w & \models [(\am,\amf{e})]\varphi\nsp&\textrm{iff }&M,w \models \I{Pre}(e) \textrm{~implies~} M \otimes \am,(\kmf{w},\amf{e}) \models \varphi
\end{align*}
\end{definition}







\subsection{Our best action-model solution for IIS}

\begin{wrapfigure}{r}{0.25\textwidth}
 \vspace{10pt}
    \centering
    \begin{tabular}{c}
$\{a\}\{b\}$\\
\begin{tikzpicture}[shorten >=0pt,node distance=1.5cm,on grid,auto,semithick,>=stealth']
  \node[state]   (q_0)                      {$a$};
  \node[state]           (q_1) [right=of q_0] {$b$};
  \path[->] (q_0) edge              node        {} (q_1);
\end{tikzpicture} \\
$\{b\}\{a\}$ \\
\begin{tikzpicture}[shorten >=0pt,node distance=1.5cm,on grid,auto,semithick,>=stealth']
  \node[state]   (q_0)                      {$a$};
  \node[state]           (q_1) [right=of q_0] {$b$};
  \path[->] (q_1) edge              node        {} (q_0);
\end{tikzpicture} \\
$\{a,b\}$ \\
\begin{tikzpicture}[shorten >=0pt,node distance=1.5cm,on grid,auto,semithick,>=stealth',bend angle=20]
  \node[state]   (q_0)                      {$a$};
  \node[state]           (q_1) [right=of q_0] {$b$};
  \path[->]
  (q_1) edge [bend left]             node        {} (q_0)
  (q_0) edge [bend left]             node        {} (q_1);
\end{tikzpicture}  \\
    \end{tabular}
    \caption{Communication graphs for two-agent IIS.}
    \label{fig:CG_IIS_2P}
 \vspace{-15pt}
\end{wrapfigure}
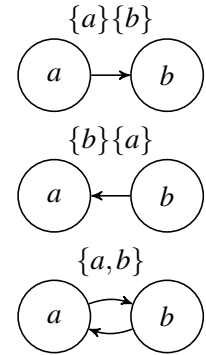




We now present our best action-model approach of modeling IIS for agents $a$ and $b$ with binary inputs.
We exploit the fact that for two-agent IIS, \emph{the epistemic models will always be bipartite graphs}.
We can hence partition the set of worlds in $M^i$ into two sets $W_1^i$ and $W_2^i$ so that any pair of distinct 
worlds in the same set can be distinguished by both agents.
For each set $W^i_j$, we use three events to represent the different sequences of concurrency classes that can happen in
a round: $\{a\}\{b\}$, $\{a,b\}$, and $\{b\}\{a\}$ (see Fig.~\ref{fig:CG_IIS_2P}).
Thus we have six events: 
three for operating with the worlds in $W^i_1$, 
and three for operating with the worlds in $W^i_2$.
The sketch of the action model is shown in Fig.~\ref{fig:am3_iis_2p_bi_r1}.
In such a sketch, the preconditions, $\phi_1$ and $\phi_2$, change from round to round. 
$\phi_j$ is a disjunction of formulas identifying the worlds in $W^i_j$.
A formula identifying a world is a conjunction of the formulas describing the local state of each agent.
In Appendix~\ref{sec:vf}, we define functions that compute an epistemic logic formula that describes the local state of an agent.
For the first round, if we consider $W_1^0=\{(0,0), (1,1)\}$ and $W_2^0=\{(0,1), (1,0)\}$, the preconditions are:
$\phi_1 = (\I{in}_{a,0} \land \I{in}_{b,0}) \lor (\I{in}_{a,1} \land \I{in}_{b,1})$, and $\phi_2 = (\I{in}_{a,0} \land \I{in}_{b,1}) \lor (\I{in}_{a,1} \land \I{in}_{b,0}).$

This approach appears to be a succinct representation of the full-information execution dynamics.
There are, however, still issues.
We would like to represent communication defined by an oblivious model \emph{just once} because the allowed communication patterns are the same regardless of the round.
All correct action models we have been able to find 
have preconditions that change from round to round. 
Moreover, the size of the formulas we get from the $\varphi$ functions grows exponentially in the number of rounds.
This suggests that in certain cases, a straightforward application of action models might not be ideal.

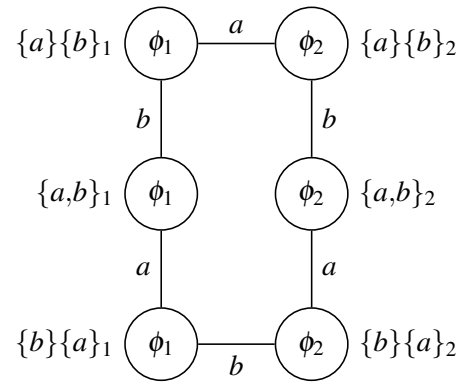
\begin{wrapfigure}{r}{0.4\textwidth}
\vspace{-15pt}
    \centering
\begin{tikzpicture}[shorten >=0pt,node distance=2cm,on grid,auto,semithick]
  \node[state,label=right:$\{a\}\{b\}_{2}$]   (q_0)                     {$ \phi_2 $};
  \node[state,label=right:$\{a\mc b\}_{2}$]           (q_1) [ below=of q_0] {$ \phi_2 $};
  \node[state,label=right:$\{b\}\{a\}_{2}$]           (q_2) [below=of q_1] {$ \phi_2 $};
  \node[state,label=left:$\{b\}\{a\}_{1}$] (q_3) [left=of q_2] {$ \phi_1 $};
  \node[state,label=left:$\{a\mc b\}_{1}$] (q_4) [above=of q_3] {$ \phi_1 $};
  \node[state,label=left:$\{a\}\{b\}_{1}$] (q_5) [above=of q_4] {$ \phi_1 $};
  \path (q_0) edge              node            {$b$} (q_1)
        (q_1)     edge              node        {$a$} (q_2)
        (q_2)          edge              node        {$b$} (q_3)
        (q_3)          edge              node        {$a$} (q_4)
        (q_4) edge              node            {$b$} (q_5)
        (q_5)          edge              node        {$a$} (q_0);
\end{tikzpicture}
    \caption{Sketch of the action model for two-agent IIS with binary inputs.}
    \label{fig:am3_iis_2p_bi_r1}
\vspace{-10pt}
\end{wrapfigure}

We have not been able to find a similar family of action models for three 
agents. 
We would need to analyze if the corresponding epistemic models are always $n$-partite, 
and how we could join all the needed events.
Finding action models for the case of $m$-ary inputs for $m \geq 3$ would be even harder.
Making things worse, the analysis might be different in distinct models:
we would need to study each model to take advantage of its own characteristics.
All these facts motivated us to look for a different and more appropriate approach.

\vspace{-0.2cm}
\section{Communication pattern models}
\vspace{-0.2cm}

Intuitively, a communication pattern model can be viewed as a non-directed graph whose vertices have two labels: a formula and a communication graph.

\subsection{Definition of communication pattern models}

First, we define our \emph{communication pattern models}.
Then, we define the syntax of our language.
After that, we define our restricted modal product.
Finally, we define our language semantics.

\begin{definition}[Communication pattern model] 
    $\cpm$ is a tuple $\cpmf{(CP,R,Pre,\overline{N})}$, where
    $\cpmf{CP}$ is a non-empty finite set whose elements are called \emph{communication patterns},
    $\cpmf{R : \I{Ag} \rightarrow \wp(CP \times CP) }$ is a function that associates each agent with an equivalence relation over the set of communication patterns,
    $\cpmf{Pre: CP \rightarrow \mathcal{L}_K}$ is a function that associates each communication pattern with a precondition, and
    $\cpmf{\overline{N}: CP \times \I{Ag} \rightarrow \wp(\I{Ag})}$ is a function that associates a (communication pattern,agent)-pair with a subset of $\I{Ag}$.
\end{definition}

We can think of communication patterns $\I{cp} \in \I{CP}$ as communication events.
The $\cpmf{\overline{N}}$ function describes the communication graph associated with a communication pattern:
$\cpmf{\overline{N}(cp,\I{a})}$ is the in-neighborhood of $a$ in such a communication graph.

\begin{definition}[Syntax]
    Let $\cpm$ be a communication pattern model over $\I{Ag}$ and $\I{Props}$.
    The language $\mathcal{L}_{\odot}$ is given by the following \textit{BNF}:
$$\varphi ::= p \mid \lnot \varphi \mid~\varphi \land \varphi~|~K_a\varphi~|~[(\cpm,\cpmf{cp})]\varphi$$
\noindent
where $p \in Props$, $a \in \I{Ag}$, $\cpmf{cp \in CP}$
and $(\cpm,\cpmf{cp})$ is an \textit{update}.
\end{definition}

%

\begin{definition}[Restricted modal product]
Let $\km=(W,\sim,L)$ be an epistemic model over $\I{Ag}$ and $\I{Props}$.
Let $\cpm \cpmf{=(CP,R,Pre,\overline{N})}$ a communication pattern model.
Let $a \in \I{Ag}$ an agent.
$(W',\sim',L') = \km' = \km \odot \cpm $ is defined as follows:

\begin{itemize}
    \item \vspace{-5pt}$W' = \{ (w,\cpmf{cp}) \in  W \times \cpmf{CP} \mid w \in W \,\land\, \cpmf{cp \in CP} \,\land\, M,w \models \I{Pre}(\cpmf{cp}) \} $
    \item \vspace{-5pt}
        $\sim'_a = \{((w,\cpmf{cp}),(w',\cpmf{cp'})) \in W' \times W' \mid w\sim_a w'\,\land\,\I{cp}\;R_a\;\I{cp}'\;\land$ \\
        \hspace*{165pt}$\underline{\cpmf{\overline{N}}(\cpmf{cp},a) = \cpmf{\overline{N}}(\cpmf{cp'},a)}\,\land$\\
        \hspace*{165pt}$\underline{w\sim_{a'} w'~\forall a' \in \cpmf{\overline{N}}(\cpmf{cp},a) }\}$
    \item \vspace{-5pt} $L'((w,\cpmf{cp}))=L(w)$
\end{itemize} 
\end{definition}

Intuitively, the first underlined condition requires agent $a$ to receive information from the same set of processes
in both $\cpmf{cp}$ and $\cpmf{cp}'$,
and the second one requires all processes in such a set to send the same information since such processes are required not to distinguish between $w$ and $w'$.

\newcommand{\ns}{\hspace*{-75pt}}
\begin{definition}[Semantics]

%
%
Let $\km=(W,\sim,L)$ be an epistemic model over $\I{Ag}$ and $\mathit{Props}$. 
Let $w,w' \in W$ be worlds. 
Let $a \in \I{Ag}$ be an agent.
Let $\cpm=\cpmf{(CP,R,Pre,\overline{N})}$ be a communication pattern model.
Let $\cpmf{cp} \in \cpmf{CP}$ be a communication pattern.
Let $\varphi,\psi \in \mathcal{L}_\odot$ be formulas.
\begin{align*}
M,w & \models p\ns&\textrm{ iff }&~p \in L(w)\\
M,w & \models \lnot \varphi\ns&\textrm{iff }&M,w \not\models  \varphi\\
M,w & \models \varphi \land \psi\ns&\textrm{iff }&M,w \models \varphi~\textrm{and}~ M,w \models \psi\\
M,w & \models K_a\varphi\ns&\textrm{iff }&M,w' \models \varphi~\textrm{for all $w'$ such that }~ w\sim_a w'\\
M,w & \models [(\cpm,\I{cp})]\varphi\ns&\textrm{ iff }& M,w \models \I{Pre}(\I{cp}) \textrm{ implies } M \odot \cpm,(w,\I{cp}) \models \varphi
\end{align*}
\end{definition}

\paragraph{Action models and communication pattern models.}
Communication pattern models are at least as general as action models.
Notice that we can build a \emph{degenerate} communication pattern model given an action model.
Let $\am \amf{=(E,R,Pre)}$ be an action model. We build a communication pattern model
$\cpm\cpmf{=(\amf{E,R,Pre},\cpmf{\overline{N}})}$ so that  $\cpmf{\overline{N}}(e,a) = \emptyset~\forall (e,a) \in \amf{E} \times \I{Ag}$.
It is easy to see that $\km \otimes \am = \km \odot \cpm$ holds.


\vspace{-0.2cm}
\subsection{Communication pattern models for arbitrary adversaries} \label{subs:cpm_adnm}

Consider any adversary $Adv$ and the initial model $M^0$ defined in Section~\ref{sec-dc-models}.
Here we define an infinite sequence $\cpm^1, \cpm^2, \hdots$ of communication pattern models
that succinctly model the evolution of knowledge in the executions of $Adv$. 
More precisely, Theorem~\ref{lemma:1} in the next section will show that
the epistemic model $\km^r=(W^r,\sim^r,L^r) = \km^0 \odot \cpm^1  \odot \cpm^2  \odot \dots  \odot \cpm^r$
captures how knowledge changes after $r$ rounds of communication.

For every $i \geq 1$, the communication pattern model $\cpm^i = \cpmf{(CP^i, R^i, Pre^i, \overline{N}^i)}$ is defined as follows:
\begin{itemize}

\item $CP^i = \{\I{cp} \in CP_{\I{Ag}} \, | \hbox{ $\exists$ an $i$-execution $(I, S \cdot \I{cp} )$ of $Adv$}\}$.

\item For every $a \in Ag$, $\cpmf{R}^i_a = \{ \cpmf{(cp,cp')} \in \cpmf{CP}^{i} \times \cpmf{CP}^{i} \mid N^{-}_{\cpmf{cp}}(a) = N^{-}_{\cpmf{cp}'}(a)\}$.

\item For every $(\cpmf{cp},a) \in \cpmf{CP}^i \times \I{Ag}$, $\cpmf{\overline{N}^{i}}(\I{cp},a) = N^{-}_{\I{cp}}(a)$.

\item For every $\I{cp} \in CP^i$,  let $\cpmf{\mathcal{W}^{i-1}_{cp}} = \{ (I,S) | \hbox{ $\exists$ an $i$-execution $(I, S \cdot \I{cp} )$ of $Adv$}\}$.
Thus, $\cpmf{Pre}^{i}(\I{cp}) = \bigvee_{(I,S) \in \mathcal{W}^{i-1}_{\I{cp}}} \varphi(I,S),$
where $\varphi(I,S) = \bigvee_{0\leq i\leq n} \varphi_i(a_i,C(i)).$ See Appendix~\ref{sec:vf} for the definition of $\varphi_i$.

\end{itemize}

\vspace{-0.3cm}
\paragraph{The case of oblivious dynamic-network models.}


\begin{wrapfigure}{r}{0.25\textwidth} 
\vspace{-10pt}
    \centering
\begin{tikzpicture}[shorten >=0pt,node distance=1.5cm,on grid,auto,semithick]
  \node[state,label=left:$\{a\}\{b\}$]   (q_0)                      {$\top$};
  \node[state,label=left:$\{a\mc b\}$]           (q_1) [below=of q_0] {$\top$};
  \node[state,label=left:$\{b\}\{a\}$]           (q_2) [below=of q_1] {$\top$};
  \path (q_0)   edge    node    {$b$} (q_1)
        (q_1)   edge    node    {$a$} (q_2);
\end{tikzpicture}
    \caption{Communication pattern model $\cpm_{\hbox{two-IIS}}$ for two-agent IIS.}
    \label{fig:am_iis_2p}
\vspace{-10pt}
\end{wrapfigure}
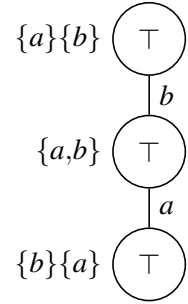

Following the definition of $\cpmf{CP}^{i}$, we can see that,
for any oblivious adversary $Adv$, $CP^i = Adv$, for each $i \geq 1$.
Thus, all $\cpm^{i}$ have the same set of communication patterns.
Moreover, for each $\cpmf{cp} \in \cpmf{CP}^i$, $\mathcal{W}^{i-1}_{\I{cp}}$ contains \emph{all} $(i-1)$-executions of
$Adv$, and hence $\cpmf{Pre}^{i}(\I{cp})$ can be set to $\top$.
Therefore, $\cpm^1 = \cpm^2 = \hdots$
The communication pattern model representing dynamics for IIS with agents $a$ and $b$ is depicted in Fig.~\ref{fig:am_iis_2p}. 
For clarity, the function $\overline{N}$ is not depicted; however, it can be obtained from the in-neighborhoods of the communication graphs in Fig.~\ref{fig:CG_IIS_2P}.
It is worth observing that the communication pattern in Fig.~\ref{fig:am_iis_2p}, omitting $\overline{N}$,
and the usual modal product $\otimes$ do not model IIS for two agents, not even for the first round. Namely, $M^1 = M^0 \otimes \cpm_{\hbox{two-IIS}}$ 
has ``undesirable'' pairs in agent relations which make $M^1$ structurally different from a 12-cycle,  which is the structure of the epistemic model for two processes with binary inputs after one round of communication in IIS (see Fig.~\ref{fig:kmIIS}).

%
%


\vspace{-0.2cm}
\subsection{The $\odot$ product reflects the change in local states through rounds}
\vspace{-0.2cm}

The dynamic epistemic logic that we present is focused on reasoning
about computations.
In particular, we are interested in modeling how
\emph{configurations} change in the full-information protocol.
A key point is that when updating an epistemic model with our modal product,
the resulting epistemic model models how the local states of agents change.
Theorem~\ref{lemma:1} below
states that
our communication pattern models do model knowledge dynamics.
The theorem formalizes this claim using the following notion.

Let $\I{Adv}$ be an adversary. For every $i \geq 0$,
we define the set 
$\mathcal{C}_{\I{Adv}}^i=\{C \mid \textrm{there is an $i$-execution $(I,S)$ of }$ $\textrm{$\I{Adv}$ that ends in the configuration }C\}.$
Let $\cpm^1, \cpm^2, \hdots\,$ be
an infinite sequence of communication pattern models.
We say that the sequence $\cpm^1, \cpm^2, \hdots\,$ \emph{reflects}
the adversary $\I{Adv}$ if for each $r \geq 1$, there is a bijection $f^{r}:
W^r \rightarrow \mathcal{C}_{\I{Adv}}^r$
such that $w\sim_{a_i} w'$ if and only if $a_i$ does not
distinguish between $f^r(w)$ and $f^r(w')$,
where $M^0$ is the initial epistemic model and $M^r=(W^r,\sim^r,L^r)=
M^0 \odot \cpm^1 \odot \cpm^2 \odot \dots
\odot \cpm^r$.
If $\cpm^1 = \cpm^2 = \hdots \,$, we simply say that $\cpm^1$ reflects $\I{Adv}$.

\begin{theorem}[Main result]
\label{lemma:1}
Let $\I{Adv}$ be an adversary
and $\cpm^1, \cpm^2, \hdots \,$  be the communication pattern models
built from $\I{Adv}$, as described in Subsection \ref{subs:cpm_adnm}.
Then, $\cpm^1, \cpm^2, \hdots\,$ reflects $\I{Adv}$.
\end{theorem}
Let $\mathcal{E}_{\I{Adv}}^r$ be the set of all $r$-executions of $\I{Adv}$.
Let $I,I'$ be two input vectors for $\I{Ag}$ and $\I{In}$.
%
%
The proof of Theorem \ref{lemma:1} will be as follows. 
First, we will present two lemmas whose proof we omit because of space restrictions.
Then, we will prove by induction that $w_r \sim_{a_i} w_r'$ if and only if $a_i$ does not
distinguish between $f^r(w_r)$ and $f^r(w_r')$.

%

Consider $E_{r+1}=(I,[\I{cp}_1,\I{cp}_2,\dots,\I{cp}_r,\I{cp}_{r+1}]) \in \mathcal{E}_{\I{Adv}}^{r+1}$, and
$E_{r}=(I,[\I{cp}_1,\I{cp}_2,\dots,\I{cp}_r]) \in \mathcal{E}_{\I{Adv}}^{r}$.
We define $g^r : \mathcal{E}_{\I{Adv}}^r \rightarrow \mathcal{C}_{\I{Adv}}^r$ as follows:
$$g^0((I,[\,]))=I.$$


$$g^{r+1}(E_{r+1})=C_{r+1}=(C_{r+1}(1),C_{r+1}(2), \dots ,C_{r+1}(n) )$$
where
$$C_{r+1}(i)(j)=\begin{cases}
    g^r(E_r)(j) & \textrm{if }a_j\in {N}^{-}_{\cpmf{cp}_{r+1}}(a_i) \cup \{a_i\} \\
    \bot &\textrm{otherwise}
\end{cases}.$$

\begin{lemma} \label{lemma:b_e_c}
    $g^r$ is a bijection. 
\end{lemma}

    Consider $w_r=(\dots((I,\cpmf{cp}_1),\cpmf{cp}_2)\dots,\cpmf{cp}_r) \in W^r$.
    We define $h^r : W^r \rightarrow \mathcal{E}_{\I{Adv}}^r$ as follows:
    $$h^r(w_r)=(I,[\cpmf{cp}_1,\cpmf{cp}_2,\dots,\cpmf{cp}_r]).$$

\begin{lemma} \label{lemma:b_w_e}
    $h^r$ is a bijection. 
\end{lemma}

Now, we start with the proof of Theorem~\ref{lemma:1}.

\begin{proof}
We define $$f^r:W^r \rightarrow \mathcal{C}^r = g^r \circ h^r.$$
Since $g^r$ and $h^r$ are bijective, $f^r$ is bijective.

Now we prove, by induction on the round number $r$, that the epistemic model $\km^r$ reflects indistinguishability between configurations.

\textbf{Base case.} 

Consider $I,I' \in W^0$, 
$C_I=f^0(I)=(I(1),I(2),\dots ,I(n))$, and $C_{I'}=f^0(I'=(I'(1),I'(2),\dots ,I'(n))$.
By construction of $\sim^0$, $I\sim^0_{a_i} I'$ if and only if $I(i) = I'(i)$ holds. 
Since $\I{a}_i$ does not distinguish between $C_{I}$ and $C_{I'}$ if and only if $I(i) = I'(i)$ holds, 
$I\sim^0_{a_i} I'$ if and only if $\I{a}_i$ does not distinguish between $C_{I}$ and $C_{I'}$ holds.

\textbf{Inductive hypothesis.}

Consider $\km^r=(W^r,\sim^r,L^r)=\km^0 \odot \cpm^1 \odot \cpm^2 \odot \dots \odot \cpm^r$, 
and $w_r,w_r' \in W^r$.
We assume that $f_{r}: W^r \rightarrow  \mathcal{C}_{\I{Adv}}^r$ satisfies that $w_r \sim^r_{a_i} w_r'$ if and only if $\I{a}_i$ does not distinguish between $f^r(w_r)$ and $f^r(w_r')$. 
    
\textbf{Inductive step.} 

%

Consider $w_{r+1}=(w_{r},cp_{r+1}),w_{r+1}'=(w_{r}',cp_{r+1}') \in W^{r+1}$.
We need to prove that $w_{r+1}\sim^{r+1}_{a_i} w_{r+1}'$ if and only if $a_i$ does not distinguish between $f^{r+1}(w_{r+1})$ and $f_{r+1}(w_{r+1}')$.

Consider $w_{r+1},w_{r+1}' \in W^{r+1}$.
By definition of $f^{r+1}$, we know that
$$f^{r+1}(w_{r+1})=C_{r+1}=(C_{r+1}(1),C_{r+1}(2), \dots ,C_{r+1}(n) )$$
where
$$C_{r+1}(i)(j)=\begin{cases} 
    f^r(w_{r})(j) & \textrm{if }a_j\in \overline{N}^{r+1}(\cpmf{cp}_{r+1},\I{a}_i) \cup \{a_i\}\\
    \bot &\textrm{otherwise}
\end{cases}$$ 
and
$$f^{r+1}(w_{r+1}')=C_{r+1}\,'=(C_{r+1}'(1),C_{r+1}'(2), \dots ,C_{r+1}'(n) )$$
where
$$C_{r+1}'(i)(j)=\begin{cases} 
    f^r(w_{r}')(j) & \textrm{if }a_j\in \overline{N}^{r+1}(\cpmf{cp}_{r+1}',\I{a}_i) \cup \{a_i\}\\
    \bot &\textrm{otherwise}
\end{cases}.$$ 

By the definition of $\odot$, $w_{r+1}\sim^{r+1}_{\I{a}_i}w_{r+1}'$ if and only if 
$w_{r}\sim^r_{\I{a}_i}w_{r}'$,
$\cpmf{cp}_{r+1}\;R^{r+1}_{\I{a}_i}\;\cpmf{cp}_{r+1}'$,\\
$\overline{N}^{r+1}(\cpmf{cp}_{r+1},\I{a}_i) = \overline{N}^{r+1}(\cpmf{cp}_{r+1}',\I{a}_i)$, and 
$w_r\sim^r_{\I{a}_j}w_{r}'\;\forall \I{a}_j \in \overline{\cpmf{N}}^{r+1}(cp_{r+1},\I{a}_i)$.


$C_{r+1}(i)(j)=\bot$ if and only if $C_{r+1}' = \bot$ holds because
by construction of $\mathcal{P}^{r+1}$, 
$cp_{r+1}\;R^{r+1}_{\I{a}_i}\;cp_{r+1}'$ if and only if $\overline{N}^{r+1}(cp_{r+1},\I{a}_i) = \overline{N}^{r+1}(cp_{r+1},\I{a}_i)$ holds.
Then,  $C_{r+1}(i)(j) = C_{r+1}(i)(j)$ holds if and only if $f^r(w_{r})(j) = f^r(w_{r}')(j)$ holds for all agents in $\overline{N}^{r+1}(cp_{r+1},\I{a}_i)$.
By the inductive hypothesis, we have that $f^r(w_{r})(j) = f^r(w_{r}')(j)\;\forall a_j \in \overline{N}^{r+1}(cp_{r+1},\I{a}_i)$ holds. 
Then, $w_{r+1}\;\sim^{r+1}_{\I{a}_i}\;w_{r+1}'$ holds if and only if $C_{r+1}(i)(j) = C_{r+1}'(i)(j)$.
$C_{r+1}(i)(j) = C_{r+1}'(i)(j)$ holds if and only if $C_{r+1}(i) = C_{r+1}'(i)$ holds. 
Hence, $C_{r+1}(i) = C_{r+1}'(i)$ holds if and only if $a_i$ does not distinguish between $C_{r+1}$ and $C_{r+1}'$.

%
%


\end{proof}

\begin{corollary}[Constant space]
Modeling an oblivious adversary $\I{Adv}$ with communication pattern models require constant space.
\end{corollary}
\begin{proof}
    Let $\cpm$ be the communication pattern model for $\I{Adv}$ built as described in Subsection \ref{subs:cpm_adnm}.
    By Theorem \ref{lemma:1}, $\cpm$ reflects $\I{Adv}$. Moreover, $\cpm$ remains the same in all rounds.
\end{proof}

\section{Related work}
\label{related-work}

The formal treatment of knowledge in distributed computing was pioneered by Halpern and Moses in~\cite{halpern-moses-84}.
Perhaps their most important result is having proved
that common knowledge amounts to simultaneity.
The book by Fagin, Halpern, Moses, and Vardi~\cite{fagin-halpern-moses-vardi-95}
was pivotal, as it summarized numerous results and compared different approaches to studying many aspects of knowledge in a system of agents.









Action models first appeared in~\cite{baltag-lawrence-slawomir-98}.
Such a formalism, however, was only considered for modeling evolution of knowledge in distributed systems, as far as we know, in~\cite{goubault-ledent-rajsbaum-18}, by Goubault, Ledent, and Rajsbaum and in~\cite{pfleger-schmid-18}, by Pfleger and Schmid.

Closer to our work is~\cite{goubault-ledent-rajsbaum-18}, where the authors exhibit a tight connection between the topological approach~\cite{herlihy-kozlov-rajsbaum-14} to distributed processing and Kripke models.
A second contribution of~\cite{goubault-ledent-rajsbaum-18}
is employing the \emph{restricted modal product} operator of action models to model knowledge change between agents after a round.
A third important result 
is employing action models to represent ``tasks''.
A task is the equivalent of a function in distributed computability. 
The task defines the possible inputs to the agents, and for each set of inputs, it specifies the set of outputs that the agents may produce. 
By representing the task itself, the possibility of solving a task amounts to the existence of a certain simplicial map. 

 The objective of~\cite{pfleger-schmid-18}, which uses action models as
well, is that of obtaining lower limits on the number of bits
necessary
for implementing a protocol that is specified with an initial
epistemic model and an infinite sequence of action models that
describe how the epistemic model
is updated through an infinite sequence of  communication rounds.
Like us,~\cite{pfleger-schmid-18} uses dynamic-network models.
Unlike us,~\cite{pfleger-schmid-18} assumes that the action model are given.
As a result,~\cite{pfleger-schmid-18} does not build an action
model and is not concerned with the size of the action models.

The work in~\cite{bjorndahl-nalls-2021} exhibits drawbacks similar to the ones we found when using the action model framework in other contexts.
The authors propose an extension of epistemic models adding a function and an update mechanism.
Adding such a function decreases the number of events needed to represent certain problems.
Our proposal, however, can be directly applied to the context of distributed systems by the communication between agents.


\section{Concluding Remarks}

The formalization of knowledge in the distributed-computing literature has still to have a more significant impact.
The evidence is that many papers in distributed computing refer to knowledge informally.

At the same time, in the epistemic-logic literature, the formalism of action models has emerged as an important mechanism for modeling the evolution of knowledge.
Hence, the works by Goubault, Ledent, and Rajsbaum~\cite{goubault-ledent-rajsbaum-18}, establishing a connection between action models and a topological approach to distributed systems, and by Pfleger and Schmid~\cite{pfleger-schmid-18}, modeling a dynamic-network protocol by an initial epistemic model and an infinite sequence of action models, are relevant. 

The approach of~\cite{goubault-ledent-rajsbaum-18} operates an action model with an epistemic model capturing knowledge at a certain point in time, to obtain a new epistemic model for knowledge after one round of communication.
We observed however, that the action models proposed in~\cite{goubault-ledent-rajsbaum-18} for IIS have certain inconveniences. 
Such action models are structurally isomorphic to the desired epistemic model, hence the number of events grows exponentially in the round number.

We proposed a family of action models with six events, for the case of two agents with binary inputs, whose preconditions change from round to round.
For obtaining such a family however, we needed to know in advance the structure of the epistemic models in further rounds.
Furthermore, the analysis for more agents or even more inputs seems to be more difficult.
Hence, a generalization of such a family is unclear for IIS.
Moreover, the analysis would depend on how the epistemic models change in different distributed-computing models.

To overcome these disadvantages, we proposed an extension of action models for dealing with communication patterns, called \emph{communication pattern models}. 
Our models work for a large variety of distributed-computing models, called dynamic-network models.
Using our extension, we were able define communication pattern models systematically for every round of execution in the full-information protocol.
In the case of oblivious models, which includes IIS, the communication pattern model remains the same all through the computation.
In either case, our approach can be applied in automated distributed-systems verification.
We emphasize the fact that communication pattern models as presented in this work are designed to deal with the full-information protocol.
We plan to analyze modifying definitions to deal with arbitrary protocols.

Communication pattern models were presented as an extension of action models.
It is possible, however, to present the same idea with a set of communication graphs.
When analyzing arbitrary dynamic-network models, there should be a precondition for each communication pattern.
When analyzing oblivious models there is no need of such precondition because it is always true.
An advantage of presenting communication pattern models as an extension of action models is that of
studying how an action model can be seen in an agent-communication perspective.

An alternative approach to modeling distributed systems epistemically is by the use of \emph{interpreted} systems, as in~\cite{halpern-moses-84},
or in the more recent papers by Casta\~neda, Gonczarowski, and Moses~\cite{castaneda-gonczarowski-moses-2014}, as well as Moses~\cite{moses-2016}.
In these works, protocols are modeled explicitly, and indistinguishability is generated directly from the local states;
consequently  there is no need for a communication pattern model (or an action model) that models the dynamics of the system.
Since we use epistemic models and communication pattern models, we need to show that the indistinguishability relation that they generate coincides with the one based on local states in the corresponding model, which is shown in Theorem~\ref{lemma:1}.
A benefit of our approach, however, is that the communication pattern models that we compute are arguably a succinct representation of the communication that can occur in a model.

\section{Acknowledgment}
We should like to thank Hans van Ditmarsch for his insightful comments.

\bibliographystyle{eptcs}
\bibliography{bibliography-delbycg}

\appendix




\section{Views and epistemic formulas} \label{sec:vf}
%
Here, we first show a way of thinking about local states in distributed computing called \emph{views}.
We then give a formal way of representing such views with an epistemic-logic formula. 

In the distributing-computing literature, it is common to regard the local states of the agents as their views.
We can think of a view of an agent as a single variable whose value changes from round to round.
Such a view takes different values depending on the round.

\begin{definition}[View]
Consider $S_{k} = [\cpmf{cp}_1,\cpmf{cp}_2,\dots,\cpmf{cp}_k]$, and $S_{k+1}$ = $S_{k} \cdot \cpmf{cp}_{k+1}$ so that $(I,S_{k+1})$ is a $k+1$-execution.
The $\I{view}$ of an agent $a_i$ in a execution $(I,S)$, $\mathit{view}(a_i,(I,S))$ for short, in the full-information protocol is defined inductively as follows:
$$\mathit{view}(a_i,(I,[\,])) = I(i).$$
$$\mathit{view}(a_i,(I,S_{k+1})) = [\mathit{view}[1],\mathit{view}[2],\dots,\mathit{view}[n]]$$
where
$$\mathit{view}[j] = 
\begin{cases} 
    \mathit{view}(a_j,(I,S_{k})) &\textrm{if~}a_j\in N^{-}_{\cpmf{cp}_{k+1}} \cup \{a_i\} \\
    \bot &\textrm{otherwise}
\end{cases}
$$ 
\end{definition}

In the full-information protocol, each agent tries to communicate its whole local state to the other agents.
If $a_i$ receives a message from $a_j$, $a_i$ will know all that $a_j$ knew in the previous round,
otherwise $a_i$ will not be able to know what $a_j$ could know.

Now, we formalize the notion of views building an epistemic logic formula for the view of $a_i$.

\begin{definition}
Let $\mathit{Views}^k$ be the set of all possible views of the agents in $\I{Ag}$ at the end of the $k$-th round.
Let $\mathit{Views}^k_i$ be the set of all possible views of the agent $a_i$ at the end of the $k$-th round.
Consider $\mathit{view} = [\mathit{view}[1],\mathit{view}[2],\dots ,\mathit{view}[n]] \in \mathit{Views}^{k+1}$. 
We define the functions $\varphi_k : P \times \mathit{Views}^k \rightarrow \mathcal{L}_K$, for all $k \in \mathbb{N}\cup \{0\}$ as follows:
$$\varphi_0(a_i,v) = in_{a_i,v}.$$
where $v \in \I{In}$. 
%
%
$$\varphi_{k+1}(a_i,view) = \bigwedge\limits_{j = 1}^{n}
\begin{cases}
    K_{a_i}(\varphi_{k}(a_j,\mathit{view}[j]))& \textrm{if}~\mathit{view}[j] \neq \bot\\
    \bigwedge\limits_{\mathit{view}' \in \mathit{Views}^k_j} \lnot K_{a_i}(\varphi_{k}(a_j,\mathit{view}')) & \textrm{otherwise}
\end{cases}$$

\end{definition}





\end{document}